\documentclass{cccg21}
\usepackage{amsmath}
\usepackage{amssymb}
\usepackage{doc}
\usepackage{epstopdf}
\usepackage{geometry}
\usepackage{graphicx}
\usepackage{hyperref}
\usepackage[numbers]{natbib}
\usepackage{setspace}
\usepackage{url}
\usepackage[font=small,labelfont=bf]{caption}
\newtheorem{definition}{Definition}

\title{Extensions of the Maximum Bichromatic Separating Rectangle Problem}

\author{Bogdan Armaselu
	\thanks{Work started as a student of Department of Computer Science, University of Texas at Dallas, {\tt barmaselub@gmail.com}}
}

\index{Armaselu, Bogdan}

\begin{document}
\thispagestyle{empty}
\maketitle

\begin{abstract}
In this paper, we study two extensions of the maximum bichromatic separating rectangle (MBSR) problem introduced in \cite{Armaselu-CCCG, Armaselu-arXiv}.
One of the extensions, introduced in \cite{Armaselu-FWCG}, is called \textit{MBSR with outliers} or MBSR-O, 
and is a more general version of the MBSR problem in which the optimal rectangle is allowed to contain up to $k$ outliers, where $k$ is given as part of the input.
For MBSR-O, we improve the previous known running time bounds of $O(k^7 m \log m + n)$ to $O(k^3 m  + m \log m + n)$.
The other extension is called \textit{MBSR among circles} or MBSR-C and asks for the largest axis-aligned rectangle separating red points from blue unit circles.
For MBSR-C, we provide an algorithm that runs in $O(m^2 + n)$ time.
\end{abstract}


\section{Introduction}

In this paper, we consider two extensions of the Maximum Bichromatic Separating Rectangle (MBSR) problem.

The MBSR problem, introduced in \cite{Armaselu-CCCG} (see also \cite{Armaselu-arXiv}), is stated as follows.
Given a set of red points $R$ and a set of blue points $B$ in the plane, with $|R| = n, |B| = m$, compute the axis-aligned rectangle $S$ having the following properties:

(1) $S$ contains all points in $R$,

(2) $S$ contains the fewest points in $B$ among all rectangles satisfying (1),

(3) $S$ has the largest area of all rectangles satisfying (1) and (2).

We call such rectangle a \textit{maximum bichromatic separating rectangle} (MBSR) or simply \textit{largest separating rectangle}.

Let $S_{min}$ be the smallest axis-aligned rectangle enclosing $R$ and discard the blue points inside $S_{min}$, as they cannot be avoided. 

The first extension of MBSR, called \textit{MBSR with outliers} (MBSR-O) or simply \textit{outliers version}, was introduced in \cite{Armaselu-FWCG},
and asks for the largest axis-aligned rectangle containing all red points and up to $k$ blue points outside $S_{min}$, where $k$ is given as part of the input.
That is, MBSR with outliers is a relaxation of condition (2) of the original MBSR problem.

In this paper, we introduce another extension of MBSR, called \textit{MBSR among circles} (MBSR-C) or simply \textit{circles version}, 
where blue points are replaced by blue unit circles, and the goal is to find the largest rectangle containing all red points and no point of any blue circle outside $S_{min}$.
Here we may also discard blue circles intersecting $S_{min}$ from consideration, as they cannot be avoided.

For both extensions, we assume that all points are in general position and that an optimal bounded solution exists, that is, there are no unbounded solutions.

The circles version is motivated by problems involving "imprecise" data, such as tumor extraction with large or imprecise cells as red or blue points,
or machine learning applications with probabilistic, rather than deterministic training data.

The outliers version can have applications in various domains.
For instance, in VLSI or circuit design, the goal is to place a hardware component (e.g., cooler) on a board with minor fabrication defects (blue points), where up to $k$ defects are acceptable to be covered by the component.
Here the red points may denote "hot spots" that must be covered or isolated from the rest of the board.

Other applications of MBSR extensions can be in machine learning, data science, or spatial databases.

\subsection{Related Work}
\label{Related Work}

Geometric separability of point sets, which deals with finding a geometric locus that separates two or more point sets whilst achieving a specific optimum criterion, is an important topic in computational geometry.
Various approaches deal with finding a specifc type of separator (e.g., hyperplane) when the points are guaranteed to be separable.
However, this is not always the case, and there is related work on weak separability, i.e., either allowing a fixed number of misclassifications or minimizing them.

Bitner and Daescu et. al \cite{Bitner} study the problem of finding the smallest circle that separates red and blue points (i.e., contains all red points and the fewest blue points). 
They provide two algorithms. The first of them runs in $O(m n \log m + n \log n)$ time and the second runs in $O(m^{1.5}\log^{O(1)}n + n \log n)$ time.
Both algorithms enumerate all optimal solutions.
Later, Armaselu and Daescu \cite{Armaselu-TCS} addressed the dynamic version of the problem, in which blue points may be inserted or removed dynamically, and provided three data structures.
The first one supports insertion and deletion queries in $O(n \log m)$ time, and can be updated in $O((m+n) \log m)$ time. 
The other two are insertion-specific (resp., deletion-specific) and allow poly-logarithmic query time, at the expense of $O(m n \log (m n))$ update time.

The problem of computing an MBSR was considered by Armaselu and Daescu.
For the case when the target rectangle has to be axis-aligned, the algorithm runs in $O(m \log m + n)$ time \cite{Armaselu-CCCG, Armaselu-arXiv}.
When the target rectangle is allowed to be arbitrarily oriented, an $O(m^3 + n \log n)$ time algorithm is given,
and they also provide an algorithm to find the largest separating box in 3D in $O(m^2 (m + n))$ time \cite{Armaselu-CCCG}.

Separability of imprecise points, where points are asscoiated with an imprecision region, has also been considered.
Note that, for the problem addressed in this paper, the blue circles can be thought of as imprecision regions.
When the imprecision regions are axis-aligned rectangles, de Berg et. al \cite{deBerg} come up with algorithms to find certain separators (with probability 1) and possible separators.
For certain separators, their algorithm runs in linear time, while for possible separators, the running time is $O(n \log n)$.

Armaselu, Daescu, Fan, and Raichel considered extensions of the MBSR problem. 
Specifically, they give an approach to find a largest rectangle separating red points from blue axis aligned rectangles in $O(m \log m + n)$ time, 
as well as an approach for the largest rectangle separaing red points from blue points with $k$ outliers in $O(k^7 m \log m + n)$ time \cite{Armaselu-FWCG}.

A very popular related problem is the one of computing the largest empty (axis-aligned) rectangle problem. 
Given a set of planar points $P$, the goal is to compute the largest $P$-empty (axis-aligned) rectangle that has a point $p \in P$ on each of its sides.
For the axis-aligned version, the best currently known bound for computing \textit{one} optimal solution is $O(n \log^2 n)$ time by Aggarwal et. al \cite{Aggarwal}.
Mukhopadhyay et. al \cite{Mukhopadhyay} solve the version where the rectangle can be arbitrarily oriented. 
They provide an $O(n^3)$ time algorithm that lists all optimal solutions.
Chaudhuri et. al \cite{Chaudhuri} prove that there can be $\Omega(n^3)$ optimal solutions in the worst case.

Nandy et. al considered the problem of finding the maximal empty axis-aligned rectangle among a given set of rectangles isothetic to a given bounding rectangle \cite{Nandy90}.
They show how to solve the problem in $O(n \log n + R)$ time, where $R$ is the number of rectangles. 
Later, they solved the version where obstacles have arbitrary orientation using an algorithm that takes $O(n \log^2 n)$ \cite{Nandy94}.

\subsection{Our Results}
\label{Our Results}

We first improve the result in \cite{Armaselu-FWCG} for the outlier version. 
Specifically, we first give a slight improvement that runs in $O(k^7 m + m \log m + n)$ time for $k$ outliers, and then a further improvement to $O(k^3 m + k m \log m + n)$ time
(which works when $k > (\log m)^{\frac{1}{4}}$).
We also solve the circles version and provide an algorithm that runs in $O(m^2 + n)$ time.

The rest of the paper is structured as follows.
In Section 2, we describe our improvements to MBSR-O, then in Section 3 we describe our algorithm for MBSR-C.
Finally, in Section 4 we draw the conclusions and list the future directions.

\section{Finding the Largest Separating Rectangle with $k$ Outliers}
\label{MBSR-O}

Given an integer $k \geq 0$, the goal is to find the largest axis-aligned rectangle enclosing $R$ that contains at most $k$ blue points in $B$.
We call this the maximum bichromatic separating rectangle with $k$ outliers (MBSR-O).

The approach in \cite{Armaselu-FWCG} is, in a nutshell, as follows. 
First, compute the smallest $R$-enclosing rectangle $S_{min}$ in $O(n)$ time. The lines defining $S_{min}$ partiton the plane into 8 regions outside $S_{min}$.
Namely, the four "side" regions $E, N, W, S$ and the four corner regions (quadrants) $NE, NW, SW, SE$.
For each region $q$, denote by $B_q$ the set of blue points inside $q$.

\begin{definition}\cite{Armaselu-FWCG}
A point $p \in B_{NE}$ \textit{dominates} another point $q \in B_{NE}$, if $x(p) > x(q)$ and $y(p) > y(q)$.
\end{definition}

\begin{definition}\cite{Armaselu-FWCG}
For each $B_q$ and for any $t$ such that $0 \leq t \leq k$, the \textit{$t$-th level staircase of $B_q$} is the rectilinear polygon formed by the blue points in $B_q$ that dominate exactly $t$ blue points in $B_q$.
\end{definition}

Note that an optimal solution contains $t$ points from $B_q$ if and only if it is sbounded by the $t$-th level staircase of $B_q$.
See figure \ref{fig:staircase} for an illustration of a staircase.

\begin{figure}[htp]
	\centering
	\includegraphics[scale=0.2]{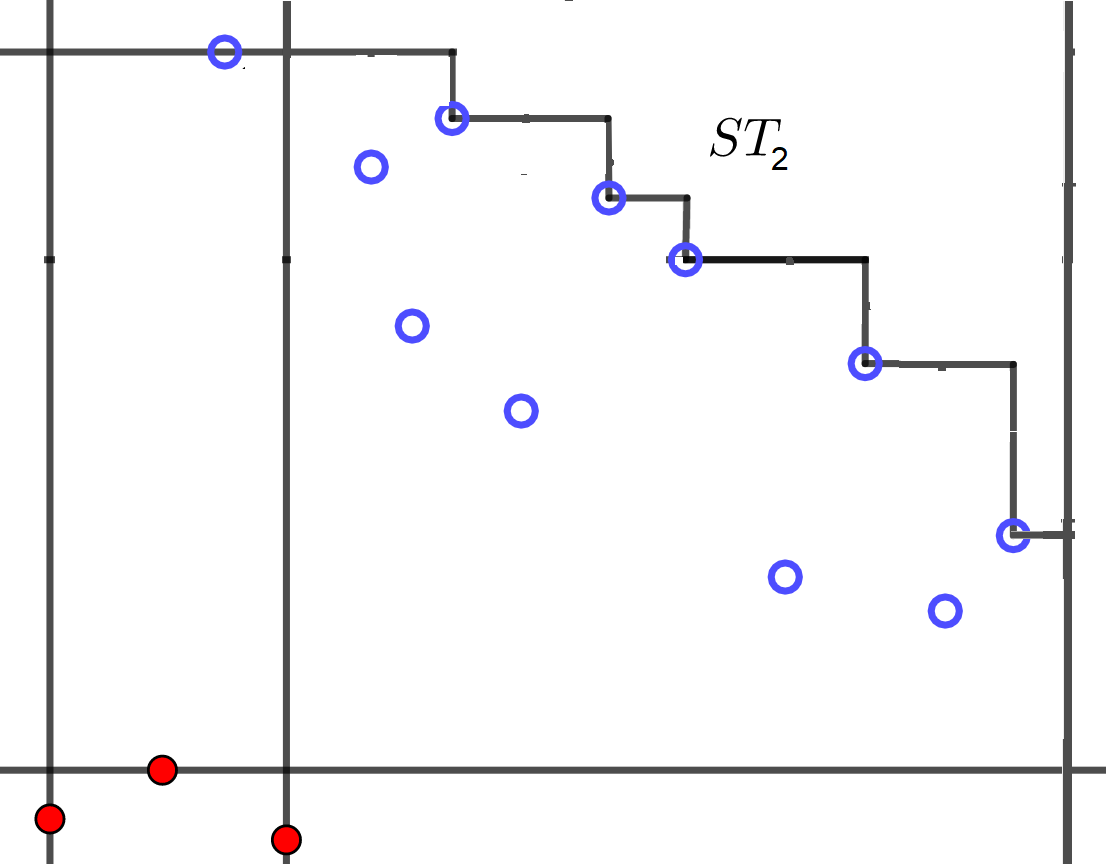}
	\caption{The 2nd level staircase of $B_{NE}$}
	\label{fig:staircase}
\end{figure}

For each partition of the number $k$ into 8 smaller integers (each corresponding to a region) $k = k_E + k_{NE} + \dots + k_{SE}$, do the following.

1. consider the $k_q + 1$-th closest to $S_{min}$ blue point from each side region $q$. These points support a rectangle $S_{max}$ which has to contain the target rectangle.

2. compute the $k_q$-level staircase $ST_{k_q}(q)$ of each quadrant $q$ in $O(m \log m)$ time.

3. Solve a "staircase" problem, i.e., find the largest rectangle containing $S_{min}$, included in $S_{max}$ and being supported by points of the staircases. 
This is done in $O(m)$ time.

There are $O(k^7)$ such partitions of $k$, so the running time of $O(k^7 m \log m + n)$ follows.

\subsection{A Slight Improvement on the Running Time}
\label{MBSR-O-improved}

To reduce the running time, we first prove the following lemma.

\begin{lemma}
\label{lemma:1.1}
The $t$-level staircases $ST_t(q)$ can be computed in $O(m \log m + m k)$ time for all $t \leq k$ and for all quadrants $q$.
\end{lemma}

\begin{figure}[htp]
	\centering
	\includegraphics[scale=0.16]{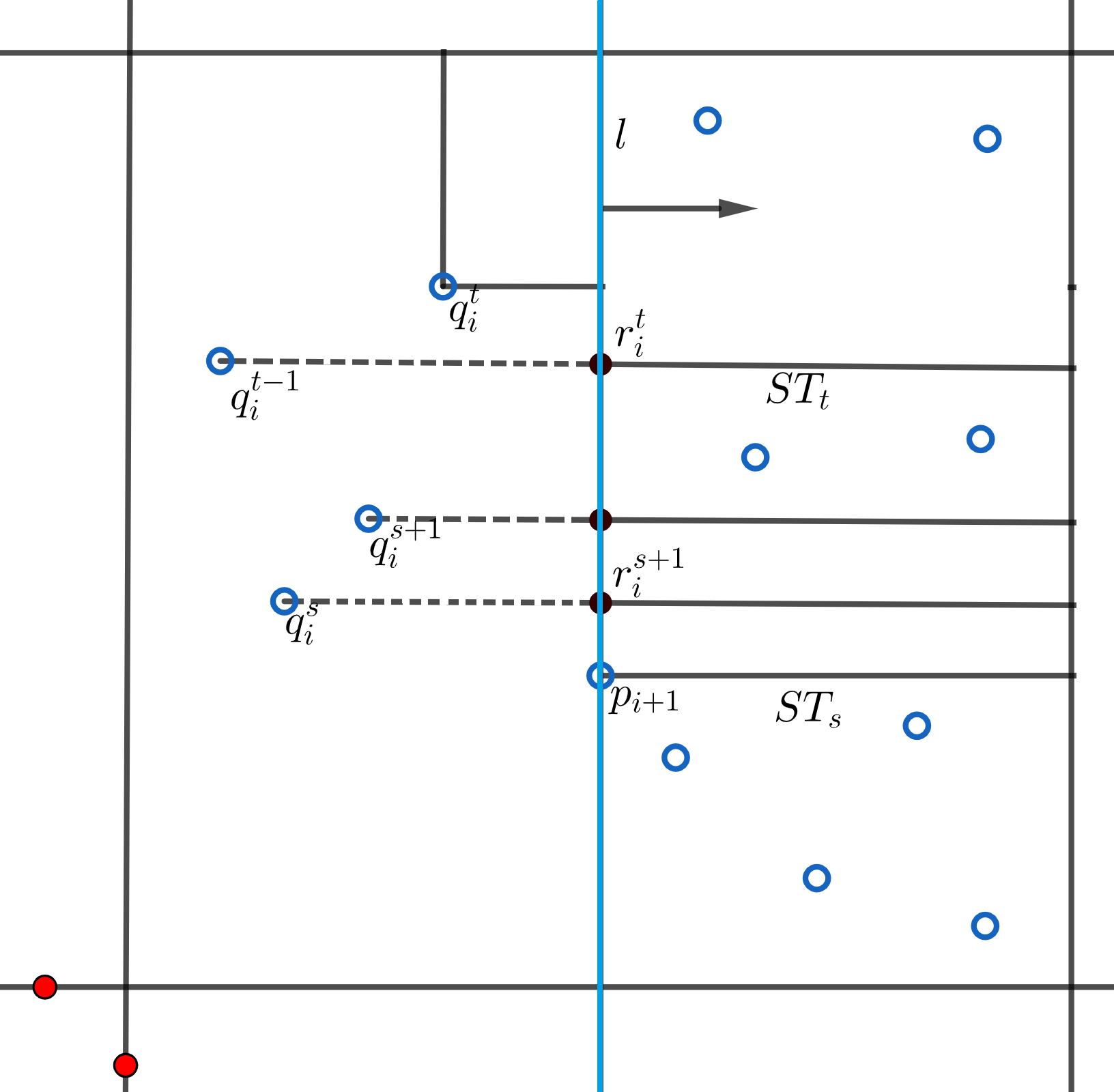}
	\caption{
		$ST_t$ is updated while sweeping vertical line $l$ over point $p_{i+1}$. 
		Since $q^s_i, \dots, q^t_i$ are higher than $p_{i+1}$, the Y coordinate of $ST_t$ is changed to that of $r^t_i$, the projection of $q^{t-1}_i$ on $l$.
		Similarly, the Y coordinates of $ST_{t-1}, \dots, ST_s$ are changed to those of $r^{t-1}_i, \dots, r^s_i$, respectively
	}
	\label{fig:st_k}
\end{figure}

\begin{proof}
We show how to compute $ST_t(NE)$ for all $t \leq k$ with a sweep line algorithm, as for other quadrants the approach is similar.
For simplicity, denote $ST_t(NE)$ as $ST_t$, that is, without specifying the quadrant.
Sort and label the points in $B_{NE}$ by increasing x-coordinate, and denote the resulting sequence as $p_1, \dots, p_m$. 
Sweep a vertical line $l$ from $x_0 = \max \{ x | (x, y) \in S_{min} \}$ to $x_1 = \infty$.
For any given position of $l$, let $P_l = \{ p_1, \dots, p_i \}$ be the set of blue points to the left of $l$. 
We maintain a balanced binary search tree $T$ over $P_l$, indexed by the y-coordinates of its elements.
The intersection of $ST_t$ with $l$ is a single point $q^t_i$, which is the highest point on $l$ that lies above at most $t$ points of $P_l$.
That is, $q^t_i$ is the $(t + 1)$-th smallest indexed entry in $T$, which we record.
As we move $l$ from left to right, $q^t_i$ can only change when $l$ intersects a point $p_i \in B_{NE}$. 
When we cross the point $p_{i+1}$ (called \textit{event point}), we insert it into $T$ and, for each $t \leq k$, we have two cases. 

1. If $p_{i+1}$ is higher than $q^t_i$, then $ST_t$ does not change height.

2. If $p_{i+1}$ is lower than $q^t_i, q^{t-1}_i, \dots, q^s_i$, for some $0 \leq s \leq t$, then $q^t_i$ is set to $q^{t-1}_i$ (which is done in $O(1)$ time), and $ST_t$ moves to the height of this entry. 
This update is repeated by setting $q^{t-1}_i$ to $q^{t-2}_i$ and so on downto $q^{s+1}_i$. 
Finally, $q^s_i$ is set to $p_{i+1}$.

If $ST_t$ changes when sweeping over $p_{i+1}$, we also record a point $r^t_i$ whose x-coordinate is that of $p_{i+1}$ and whose y-coordinate is that of the updated $q^t_i$. 
Let $Q$ be the set of all such points $q^t_i, r^t_i$ recorded during this process for all $t \leq k, i = 1, \dots, m$.
It is not hard to argue that $\cup_{t \leq k} ST_t \subseteq Q$. 
Moreover, $|Q| \leq m$ holds, as a point is added to $Q$ only when the sweep line crosses a point of $B_{NE}$. 
Finally, note that we encountered $O(m)$ event points $p_i$. 
For each of them, we spend $O(\log m)$ time to insert them into $T$ and $O(k)$ time for $q^t_i, r^t_i$ pointer updates.
Thus, the running time bound follows.
\end{proof}

Figure \ref{fig:st_k} illustrates the proof of Lemma \ref{lemma:1.1}.

Rather than computing the $ST_t$'s for each partition of $k$, we compute them all in $O(m \log m + m k)$ time before considering such partitions.
Then, for each of the $O(k^7)$ partitions, we need only solve a staircase problem in $O(m)$ time.
This gives us the following result.

\begin{theorem}
\label{theorem:1.2}
Given two sets $R, B$, with $|R| = n, |B| = m$, the largest rectangle enclosing $R$ and containing at most $k$ points in $B$ can be found in $O(k^7 m + m \log m + n)$ time.
\end{theorem}

\subsection{A Closer Look at the Number of Candidate Partitions of $k$}
\label{MBSR-O-closer-look}

In the previous section, we reduced the running time by a factor of $\log m$. 
However, it seems hard to further improve this bound given the high number of partitions of $k$.
In this subsection, we show how to reduce the running time by reducing the number of candidate partitions of $k$.

To do that, we first compute all the $t$-level staircases $ST_t, 0 \leq t \leq k$ as described in the previous section.
We then consider the blue points in 4 pairs of adjacent regions, e.g., $N$ and $NE$.
That is, we suppose the total number of outliers coming from $B_N \cup B_{NE}$, denoted $k_{NNE} = k_N + k_{NE}$, is fixed.
Similarly, we supose $k_{ESE} = k_E + k_{SE}, k_{SSW} = k_S + k_{SW}, k_{WNW} = k_W + k_{NW}$ are fixed.
Let $ST(Q) = \cup_{t=1}^{k}{ST_t(Q)}$, for any quadrant $Q$.
From now on, we focus on the $N$ and $NE$ regions and, for simplicity, we denote $ST(NE)$ as simply $ST$ and $ST_t(NE)$ as simply $ST_t$.

We notice that, even though any points of any $t$-th level staircase, $t \leq k_{NE}$, may be a corner for a candidate rectangle,
most of these rectangles can be discarded as they are guaranteed to be smaller than the optimal rectangle.

For every pair $P = (P_1, P_2)$ of regions and every integer $t: 0 \leq t \leq k$, denote by $S^P_t$ the set of pairs $(p, q) \in (B_{P_1} \cup B_{P_2})^2$ that may define an optimal solution, with $p$ as top support and $q$ with right support,
among all rectangles containing $t$ blue points from $B_{P_1} \cup B_{P_2}$.
From now on, whenever understood, we are going to remove the superscript and simply write $S_t$, e.g., $S_{k_{NNE}}$ instead of $S^{NNE}_{k_{NNE}}$.
For every $t$, we store $S_t$ as an array.

\begin{figure}[htp]
	\centering
	\includegraphics[scale=.2]{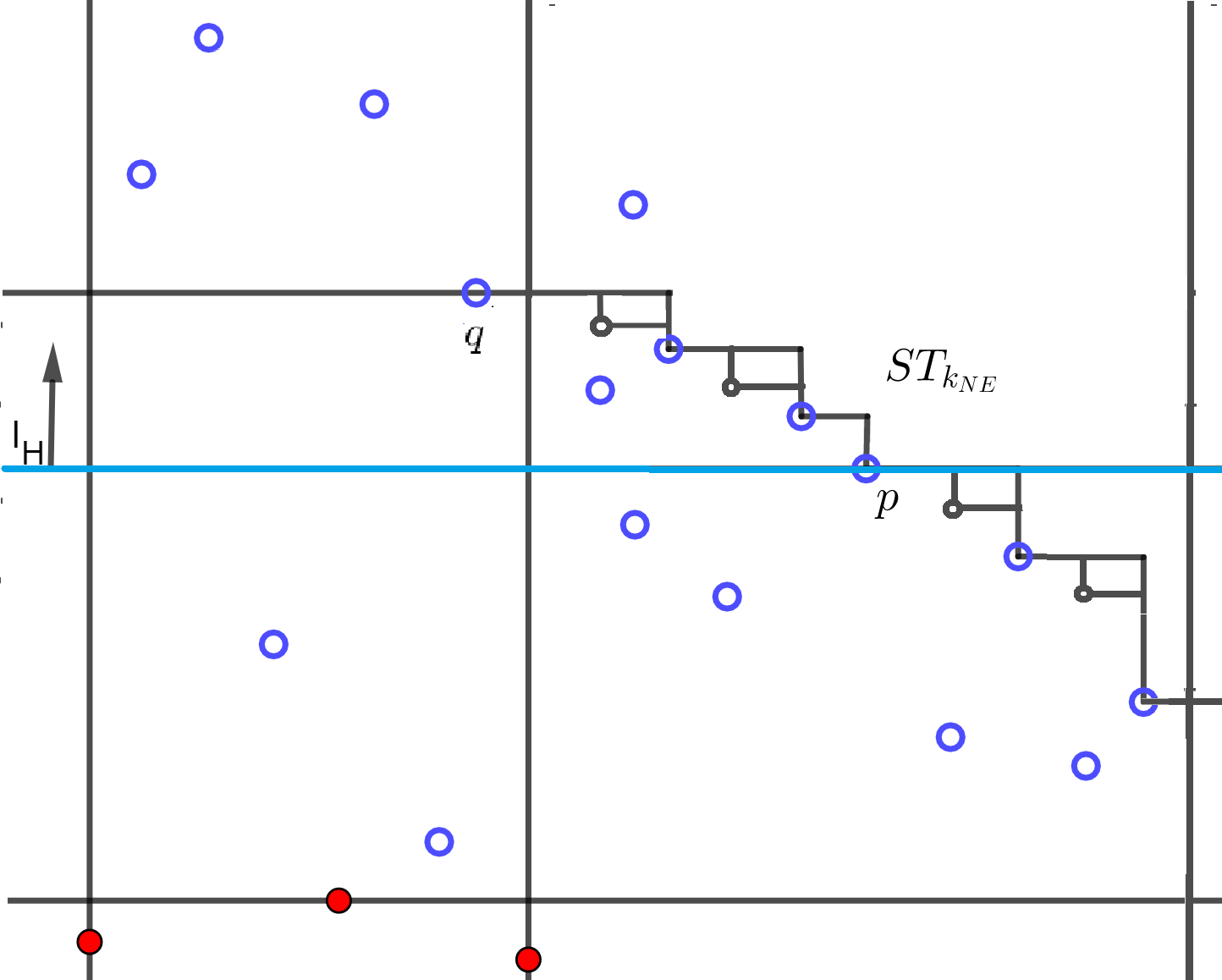}
	\caption{
		$B_N \cup B_{NNE}$ is swept with a horizontal line $l_H$ sliding upwards from the lowest blue point.
		At each blue point $p$ encountered, the set $S_{k_{NNE}}$ is updated.
		Black dots denote staircase points that are not blue points.
	}
	\label{fig:horizontal-sweep}
\end{figure}

We extend the above definition to points $p, q \in B_q$ for other quadrants $q$, by flipping inequalities accordingly.

The goal is to compute $S_{k_{NNE}}$.
Suppose we have already computed all $S_t: t < k_{NNE}$.
Sweep $B_N \cup B_{NE}$ with a horizontal line $l_H$ going upwards, starting at the $k_{NNE} + 1$-th lowest blue point in $B_N \cup B_{NE}$, as shown in Figure \ref{fig:horizontal-sweep}.
For every blue point $p$ encountered as top support, let $below(p)$ be the highest point in $B_N$ below $p$, and $above(p)$ be the lowest point in $B_N$ above $p$.
For every $p \in B_{NE}$, let $rb(p)$ be the leftmost point in $B_{NE}$ to the right of $p$ and below $p$.
Let $t_N$ be the blue point count below $p$ from $B_N$. 
Also let $t$ be the number of points dominated by $p$ from $B_N \cup B_{NE}$.

First, assume $p \in B_{NE}$ and let $t_{NE} = t - t_N$ be the number of points dominated by $p$ from $B_{NE}$, i.e., $p \in ST_{t_{NE}}$.
When sweeping the next blue point, say $q$, we do the following.

\textbf{Case 1.}
If $q$ is to the right of $p$, then $q$ is below $above(p)$ but dominates $p$, the points domiated by $p$, and the points in $T(p, q) = \{s \in B_{NE} \cap ST_{t_{NE}} | x(p) < x(s) < x(q)\}$  (Figure \ref{fig:outlier-case1}).
If $t + t_{pq} = k_{NNE}$, then we add $(q, rb(q))$ to $S_{t + t_{pq}}$, where $t_{pq} = 1 + |T(p, q)|$.

\begin{figure}[htp]
	\centering
	\includegraphics[scale=.2]{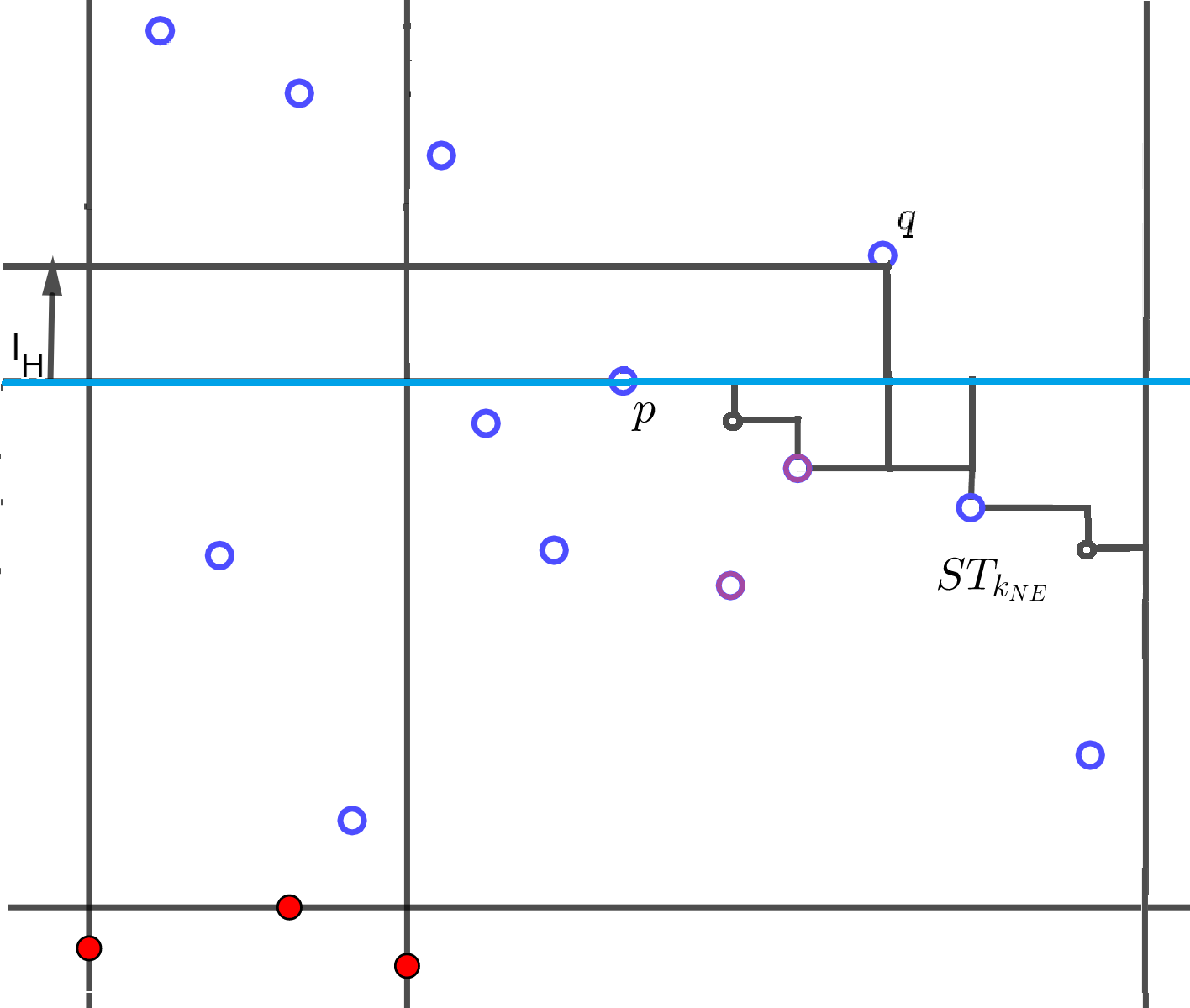}
	\caption{
		Case 1. $p \in B_{NE}$, $q$ to the right of $p$. The purple empty dots denote $T(p, q)$.
	}
	\label{fig:outlier-case1}
\end{figure}

\textbf{Case 2.}
If $q \in B_{NE}$ and $q$ is to the left of $p$, then $q$ is below $above(p)$ but dominates the points dominated by $p$, except the ones in $U(q, p) = \{s \in B_{NE} | x(q) < x(s) < x(p), y(s) < y(p)\}$.  
For each $s \in U(q, p)$, if $t - i(s) = k_{NNE}$, then we add $(q, s)$ to $S_{t - i(s)}$, where $i(s)$ is the index of $s$ in $U(q, p)$ in decreasing order of X coordinates.
Finally, if $t = k_{NNE}$, then we add $(q, p)$ to $S_t$.


\textbf{Case 3.}
If $q \in B_N$ then, for each $s \in U(p) = \{s \in B_{NE} | x(s) < x(p), y(s) < y(p)\}$ such that $t - i(s) = k_{NNE}$, we add $(q, s)$ to $S_{t - i(s)}$, where $i(s)$ is the index of $s$ in $U(p)$ in decreasing order of X coordinates.
Finally, if $t = k_{NNE}$, then we add $(q, p)$ to $S_t$. 


Now assume $p \in B_N$ and let $t_{NE}$ be the largest $t'$ such that all points in any $ST_{t'}$ are below $p$.
Let $b(p)$ be the leftmost point of $ST_{t_{NE}}$ below $p$.
When sweeping the next blue point, say $q$, we do the following.

\textbf{Case 4.}
If $q$ is to the right of $b(p)$ then, if $t = k_{NNE} - 1$, we add $(q, rb(q))$ to $S_{t+1}$.
For each $s \in U(q)$ such that $t - i(s) = k_{NNE}$, we also add $(q, s)$ to $S_{t - i(s)}$. 


\textbf{Case 5.}
If $q \in B_{NE}$ and $q$ is to the left of $b(p)$, then, if $t = k_{NNE} - 1$, we add $(q, s)$ to $S_{t+1}$ for every $(p, s) \in S_t$. 


\textbf{Case 6.}
If $q \in B_N$, then, if $t = k_{NNE} - 1$, we add $(q, s)$ to $S_{t+1}$ for every $(p, s) \in S_t$. 


The following lemma puts an upper bound on the storage required by $S_{k_{NNE}}$.

\begin{lemma}
\label{lemma:1.3}
$|S_{k_{NNE}}| = O(m)$.
\end{lemma}

\begin{proof}
In case 1, we only one pair to $S_{k_{NNE}}$. 
In case 2, even though we consider $|U(q, p)|$ points, we only add the pair $(q, s)$ such that $t - i(s) = k_{NNE}$.
Similarly, in cases 3 and 4 we only add the pair $(q, s): t - i(s) = k_{NNE}$, even though we consider $|U(p)|$ (resp., $|U(q)|$) points.
In case 5, we add at most $|ST_{t_{NE}}|$ pairs if $t = k_{NNE} - 1$.
However, note that for the subsequent point $q'$ swept, we would have a larger number$t'$ of blue points in $B_N \cup B_{NE}$ dominated by $q'$.
Thus, we only add at most $|ST_{t_{NE}}| = O(m)$ pairs once.
Similarly, in case 6 we only add $O(m)$ pairs once.
\end{proof}

This lemma bounds the running time of the aforementioned sweeping approach.

\begin{lemma}
\label{lemma:1.4}
For any $t: 0 \leq t \leq k$, the horizontal line sweeping desccribed above takes $O(m \log m)$ time.
\end{lemma}

\begin{proof}
We store the blue points in $B_N \cup B_{NE}$ in two balanced binrary search trees $X, Y$, indexed by X (resp., Y) coordinates.
Thus, for each blue point $p$ swept, we require $O(\log m)$ time.
We require an extra $O(\log m)$ time to compute $above(p), below(p)$, and $rb(p)$.
In case 1, note that we can compute $t_{pq}$ by finding the position of $q$ in the X-sorted order of $ST_{k_{NE}}$, and thus the number of blue points $s: x(p) < x(s) < x(q)$, 
in $O(\log m)$ time, since $ST_{k_{NE}}$ is maintained as a binary search tree.
Thus, we only require an extra $O(\log m)$ time to handle case 1.
In cases 2 and 4, note that we only need to add $(q, s)$ to $S_{k_{NNE}}$ if $i(s) = t - k_{NNE}$, so we query for $s$ in $X$ using $O(\log m)$ time.
Similarly, in case 3 we only add $(q, p)$ to $S_{k_{NNE}}$ if $i(s) = t - k_{NNE}$, so we query $X$ for $s$ in $O(\log m)$ time.
Now in cases 5 and 6 we spend $O(m)$ time to traverse $S_t$, since we store $S_t$ as an array for any $t$,
but they only occur once, so this gives us $O(m)$ total time.
In every case, since $S_t$ is an array, adding a pair to $S_t$ takes $O(1)$ time.
Since we sweep $O(m)$ blue points, the result follows.
\end{proof}

\textbf{Corollary}. We compute $S_t$ in $O(k m \log m)$ time for all $t: 0 \leq t \leq k$.
This holds for any quadrant.

We reduce the number of candidate partitions of $k$ from $O(k^7)$ to $O(k^3)$ as follows.
By writing $k = k_{NNE} + k_{ESE} + k_{SSW} + k_{WNW}$, we can deduce $k_{NNE} = k - k_{WNW} - k_{ESE} - k_{SSW}$ for every combination of $k_{WNW}, k_{ESE}, k_{SSW}$.
Therefore, there are $O(k^3)$ such combinations.

Initially, we compute $S^Q_t$ for every quadrant $Q$ and $0 \leq t \leq k$.
Then, for each combination $(k_1, k_2, k_3), k_1, k_2, k_3 = 0, \dots k, k_1 + k_2 + k_3 \leq k$, we set $k_4 = k - k_1 - k_2 - k_3$ 
and solve the staircase problem in \cite{Armaselu-CCCG} with the pairs $S^{WNW}_{t_1} \cup S^{ESE}_{t_2} \cup S^{SSW}_{t_3} \cup S^{NNE}_{k_4}$ as pairs of supports.
Each staircase problem takes $O(m)$ time to solve, so we require $O(k^3 m)$ time for all candidate partitions of $k$.
Putting this together with the result in Lemma \ref{lemma:1.1}, we get the following result.

We obtain the following result.

\begin{theorem}
\label{theorem:1.5}
Given two sets $R, B$, with $|R| = n, |B| = m$, the largest rectangle enclosing $R$ and containing at most $k$ points in $B$ can be found in $O(k^3 m + k m \log m + n)$ time.
\end{theorem}

\section{Finding the Largest Axis Aligned Rectangle Enclosing $R$ and avoiding Unit Circles}
\label{MBSR-C}

In this version, called the \textit{circles} version (MBSR-C), $B$ consists of unit circles that do not intersect $S_{min}$, 
and the goal is to find the largest axis-aligned rectangle enclosing $R$ that avoids all circles.

\begin{figure}[htp]
	\centering
	\includegraphics[scale=0.35]{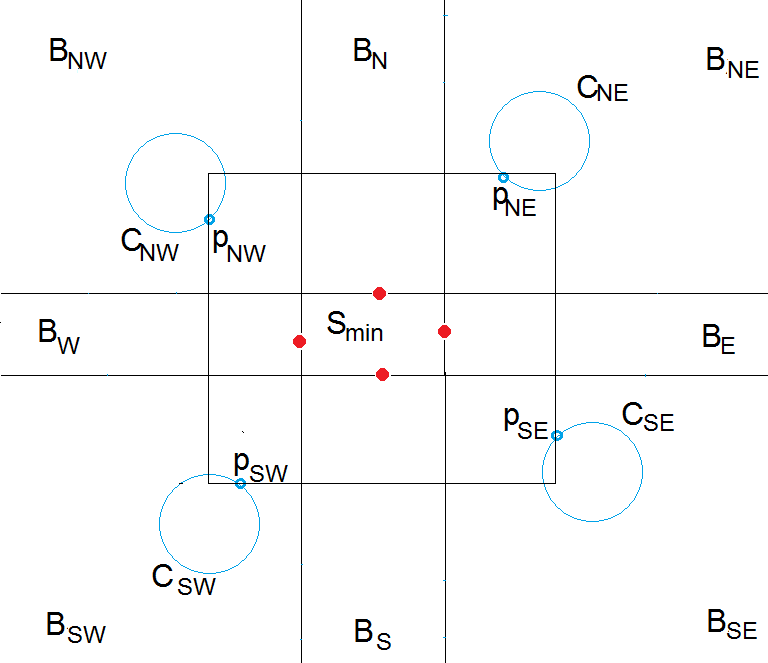}
	\caption{If the maximum rectangle separating $R$ and $B'$ were bounded by point $p \in C$, it would intersect $C$}
	\label{fig:circles}
\end{figure}

Note that the reduction in \cite{Armaselu-FWCG} for finding the largest separating rectangle among rectangles does not work.
To see why this is the case, let $C_{NW}, C_{SW}, C_{SE}, C_{NE}$ be circles in the regions $B_{NW}, B_{SW}, B_{SE}, B_{NE}$. 
If one picks any point $p$ on the quadrant of $C_{NW}$ that is the closest to $S_{min}$, and adds it to $B'$ (and does the same for $C_{SW}, C_{SE}, C_{NE}$), 
then, if the maximum rectangle separating $R$ and $B'$ is top-bounded or right-bounded by $p$, 
it intersects $C_{NW}$ (as shown in Figure \ref{fig:circles}).

\begin{figure}[htp]
	\centering
	\includegraphics[scale=0.4]{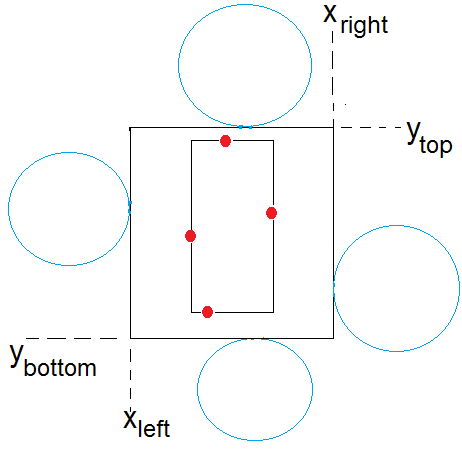}
	\caption{A circle bounding the rectangle at an edge fixes that edge in terms of X or Y coordinate}
	\label{fig:circle-edge}
\end{figure}

A \textit{candidate separating rectangle} (CSR) is a rectangle that encloses $R$ and cannot be extended in any direction without intersecting some circle.
Notice that a CSR may touch a circle either at an edge or at a corner.
If it is bounded at an edge, then that edge is fixed in terms of X or Y coordinate and the arc it touches at each endpoint of the edge is uniquely determined (Figure \ref{fig:circle-edge}).
On the other hand, if it is bounded at a corner, then the corner can be slided along the appropriate arc of the circle (Figure \ref{fig:circle-corner}).
Each position of the corner determines the X or Y coordinates of its two adjacent edges, and thus the arcs pinning the two adjacent corners, if any.

We say that an edge $e$ of a CSR is \textit{pinned} by a circle $C$, if $C$ touches the interior of $e$.

A horizontal (resp., vertical) edge $e$ is said to be \textit{fixed} by two circles $C_1, C_2$ in terms of Y (resp., X) coordinate, if:

(1) the ends of $e$ are on $C_1$ and $C_2$, respectively, and

(2) changing its Y (resp., X) coordinate would result in either $e$ intersecting $C_1$ or $C_2$ or failing to touch both $C_1$ and $C_2$.


\begin{figure}[htp]
	\centering
	\includegraphics[scale=0.4]{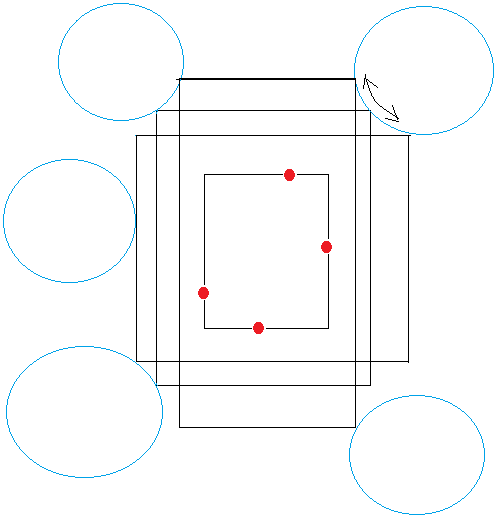}
	\caption{A circle bounding the rectangle at a corner allows the corner to slide along the arc. 
		Each position of the corner determines its two adjacent edges}
	\label{fig:circle-corner}
\end{figure}

\subsection{A description of all cases in which a CSR can be found}
\label{MBSR-C-cases}

Based on the number of edges of a CSR pinning by circles, we consider the following cases.

\begin{figure}[htp]
	\centering
	\includegraphics[scale=0.4]{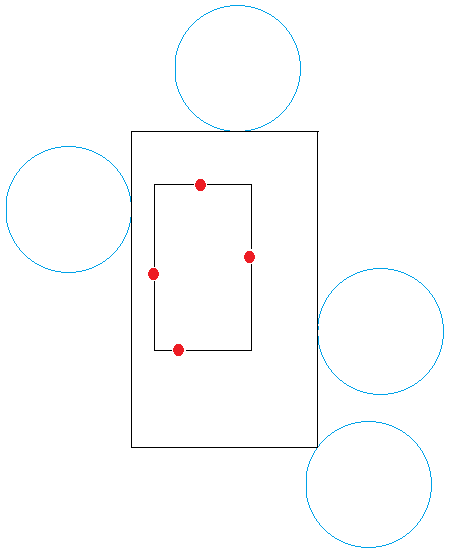}
	\caption{Case 1
}
	\label{fig:circle-case1}
\end{figure}

\textbf{Case 1}. Three edges pinned by circles (Figures \ref{fig:circle-case1}). 
In this case, we extend the the fourth edge outward from $S_{min}$ until it touches a circle.
Thus, the CSR is uniquely determined.

\textbf{Case 2}. Two edges are pinned by two circles $C_1, C_2$.
Here we distinguish the following subcases.

\begin{figure}[htp]
	\centering
	\includegraphics[scale=0.4]{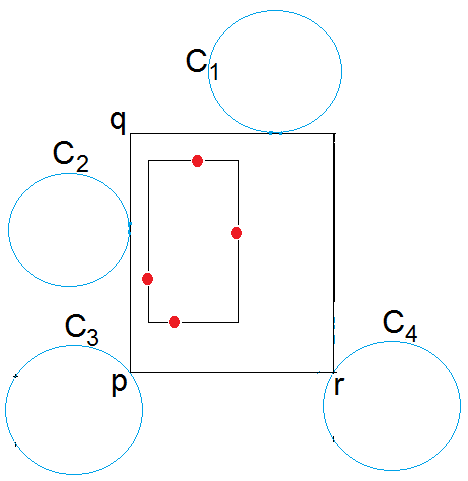}
	\caption{Case 2.1
}
	\label{fig:circle-case2-1}
\end{figure}

\textbf{Case 2.1}. Two adjacent edges are pinned by $C_1, C_2$.
Note that their common corner $q$ is fixed.
We extend one of the edges by moving its other end $p$ away from $q$, until it touches a circle $C_3$,
and then extend the third edge until it touches a circle $C_4$ at a point $r$ (Figure \ref{fig:circle-case2-1}).
The resulting CSR is unique.


\textbf{Case 2.2}. Two adjacent edges are pinned by $C_1, C_2$.
We extend one of the edges by moving its other end $p$ away from $q$, until the orthogonal line through $p$ touches a circle $C_3$ at a point $r$. 
While moving $p$, the point $r$ can slide along one or more circles in the same quadrant, giving an infinite number of CSRs.


\textbf{Case 2.3}. Two opposite edges are pinned by $C_1, C_2$.
We slide the other two edges outward from $S_{min}$ until each of them touches some circle. 
This gives us a unique CSR.

\textbf{Case 3}. One edge $e$ is pinned by a circle $C_1$.
We have the following subcases.

\begin{figure}[htp]
	\centering
	\includegraphics[scale=0.4]{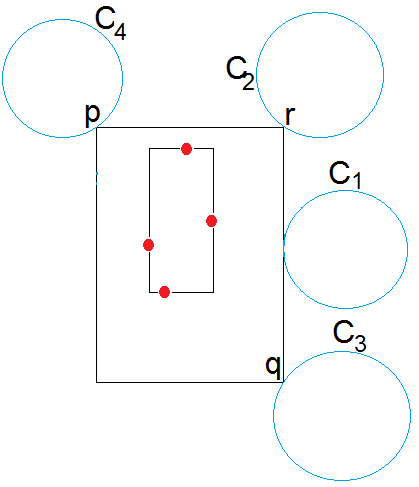}
	\caption{Case 3.1
}
	\label{fig:circle-case3-1}
\end{figure}

\textbf{Case 3.1}. When $e$ is extended in both directions, it touches two circles $C_2, C_3$ (Figure \ref{fig:circle-case3-1}). 
We then slide the fourth edge outward from $S_{min}$ until it touches a circle $C_4$ and we have a unique CSR.


\textbf{Case 3.2}. When $e$ is extended in both directions, the orthogonal line through one of the ends $p$ touches a circle $C_2$ at point $q$. 
While moving $p$, $q$ can slide along one or more circles in the same quadrant, yielding an infinite number of CSRs.
After establishing the position of $q$, we slide the fourth edge away from $S_{min}$ until it touches a circle $C_3$.

\textbf{Case 4}. No edge is pinned by any circle.
In this case, all corners can slide along circles until one of the edges becomes pinned by some circle, giving an infinite number of CSRs.
Suppose the position of a corner $p$ along a circle $C_1 \in B_{NE}$ is known.
We consider the following subcases.

\begin{figure}[htp]
	\centering
	\includegraphics[scale=0.4]{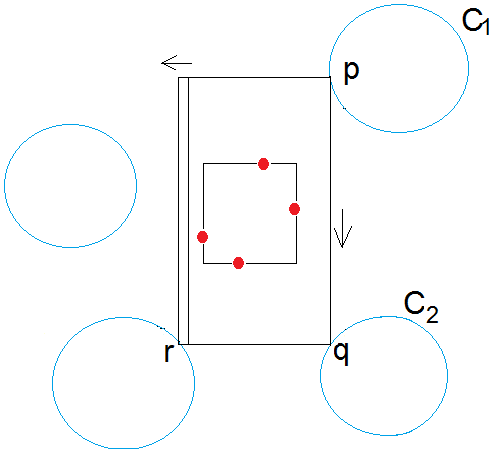}
	\caption{Case 4.1
}
	\label{fig:circle-case4-1}
\end{figure}

\textbf{Case 4.1}. While extending the CSR in the two directions away from $p$, the CSR touches a circle in $B_{SE}$ or $B_{NW}$ at some point $q$ before touching any circle in $B_{SW}$ (Figure \ref{fig:circle-case4-1}).
The other two corners are determined by sliding the edge opposite to $pq$ outwards until it touches a circle at some point $r$.
In this case, the CSR is uniquely defined.


\textbf{Case 4.2}. While extending the CSR in the two directions away from $p$, the first circle that CSR touches, at a point $q$, is located in $B_{SW}$. 
This gives us an infinite number of CSRs.

\subsection{Dominating envelopes}
\label{AAR-C-dominating-envelopes}

\begin{definition}
For each quadrant $B_i$, the \textbf{dominating envelope} of $B_i$ is a curve $\mathcal{C}$ with the following properties:

1) $\mathcal{C}$ lies inside $B_i$;

2) for any point $p \in \mathcal{C}$, the rectangle cornered at $p$ and the closest corner of $S_{min}$ from $p$ is empty;

3) extending $\mathcal{C}$ away from $S_{min}$ violates property (2).
\end{definition}

\begin{figure}[htp]
	\centering
	\includegraphics[scale=0.6]{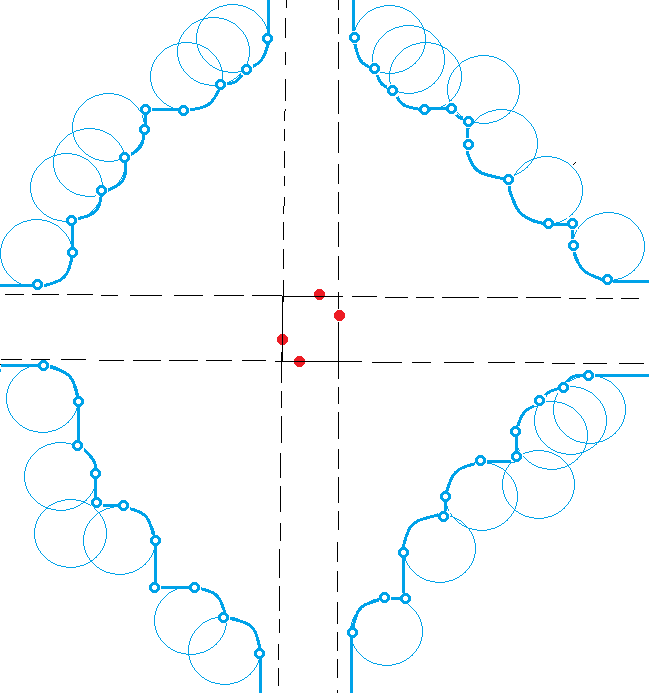}
	\caption{In each region, the circles define a dominating envelope $\mathcal{C}$, which is a sequence of arcs and horizontal or vertical segments. 
Two consectutive arcs or segments define a breakpoint on $\mathcal{C}$.}
	\label{fig:dominating-envelope-1}
\end{figure}

Note that $\mathcal{C}$ is a sequence of circle arcs, horizontal, and vertical segments, plus a horizontal and a vertical infinite ray (see Figure \ref{fig:dominating-envelope-1}).
Its use will be revealed later on.

The dominating envelope changes direction at \textit{breakpoints}, which can be between two consecutive arcs, segments, or infinite ray.
Every two consecutive breakpoints define a \textit{range of motion} for a CSR corner.

A breakpoint $p$ is said to be a \textit{corner breakpoint}, if a CSR cornred at $p$ cannot be extended away from $S_{min}$ in all directions without crossing some circle, even if its other corners are not located on any envelope.

To compute the dominating envelope $\mathcal{C}$ of $B_{NE}$, we do the following.
First, sort the circles by X coordinate of their centers.
Let $p$ be the current breakpoint (initially, the first breakpoint is the left endpoint $l$ of the left-most circle, with a vertical infinite ray upwards from $l$).
For each two adjacent circles $C_1(c_1, 1), C_2(c_2, 1)$, depending on the relative positions of $c_1, c_2$, we do the following.

\begin{figure}[htp]
	\centering
	\includegraphics[scale=0.4]{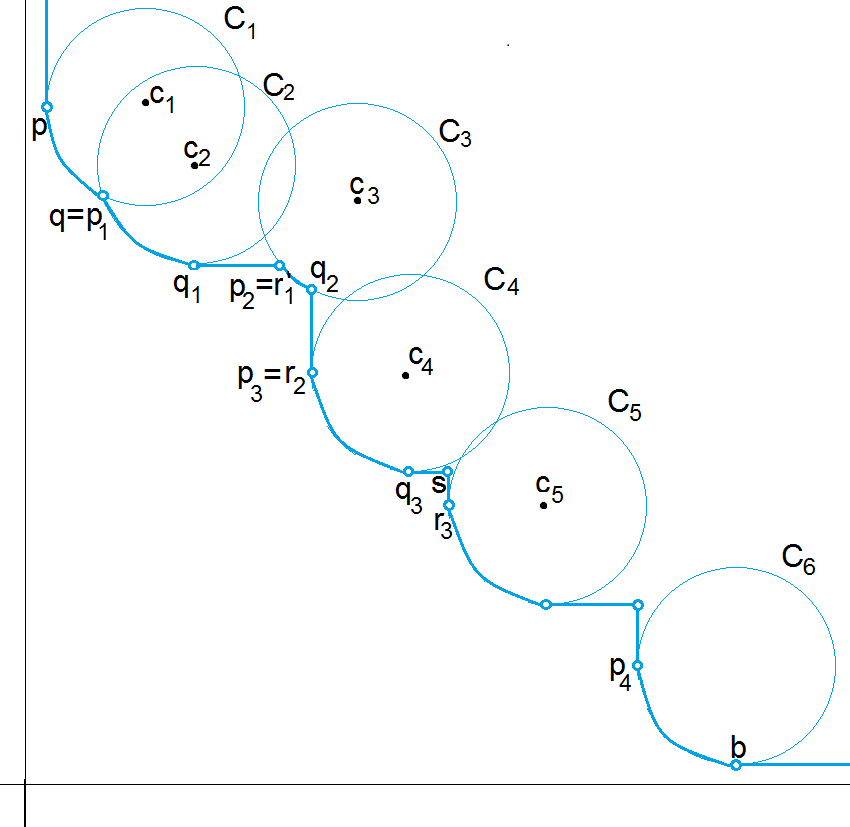}
	\caption{
	Each pair of adjacent circles gives a different case.
}
	\label{fig:dominating-envelope-2}
\end{figure}

\textbf{Case A}. $c_2 \in C_1$.
In this case, we add the lower intersection between $C_1$ and $C_2$ as a new corner breakpoint $q$, along with the arc $pq$ of $C_1$.

\textbf{Case B}. $y(c_1) - 1 \leq y(c_2) \leq y(c_1)$ and $x(c_2) > x(c_1) + 1$.
We add a breakpoint $q$ at the bottom of $C_1$, the arc $pq$ of $C_1$, a corner breakpoint $r$ at the intersection between the horizontal through $q$ and $C_2$, and the line segment $qr$.

\textbf{Case C}. $x(c_1) + 1 \leq x(c_2) \leq x(c_1) + 1$ and $y(c_2) < y(c_1) - 1$.
We add a breakpoint $r$ at the left endpoint of $C_2$, a corner breakpoint $q$ at the intersection between the vertical through $r$ and $C_1$, the line segment $qr$, and the arc $pq$ of $C_1$.

\textbf{Case D}. $x(c_2) > x(c_1) + 1$ and $y(c_2) < y(c_1) - 1$.
We add the breakpoints $q$ at the bottom of $C_1$, $r$ at the left end of $C_2$, the \textit{corner breakpoint} $s$ at the intersection between the horizontal through $q$ and the vertical through $r$, the arc $pq$ of $C_1$, and the segments $qs$ and $sr$.
%
Figure \ref{fig:dominating-envelope-2} illustrates this process.

Note that deciding the case a circle belongs to can be done in $O(1)$ time per circle.

\subsection{Finding an optimal solution in each case}
\label{MBSR-C-find-each-case}

To find an optimal solution in each case, we do the following.

First, slide the edges of $S_{min}$ outward until each of them touches a circle. 
If this is not possible, then the solution is unbounded, so we will assume that each edge will eventually hit a circle.
Denote the resulting rectangle by $S_{max}$ and discard the portions outside $S_{max}$ from the dominating envelopes of all quadrants.
The endpoints of the resulting envelopes are also counted as breakpoints.

For each case described in Section \ref{MBSR-C-cases}, we give a different algorithm to compute the largest separating rectangle.
Before considering any case, we sort the blue circles by the X coordinate and then (to break ties) by the Y coordinate of their centers.

\textbf{Case 1}.

We consider all corner breakpoints that are defined by pairs of adjacent circles in Case D, and add them to a set $B'$.
We also consider all arcs $p q$ of circles that are part of pairs in Case D ($p$ to the west of $q$), the vertical line $l_p$ through $p$ and the horizontal line $l_q$ through $q$,
and add $l_p \cap l_q$ to $B'$.
We then find the largest rectangle $S^*$ enclosing $R$ and containing the fewest points in $B'$ using the algorithm in \cite{Armaselu-arXiv} in $O(m + n)$ time.
It is easy to check that $S^*$ is an optimal solution for Case 1, since any circle containing points in $B'$ intersects a circle in $B$.
Thus, Case 1 can be done in $O(m + n)$ time.

\textbf{Case 2}. 

Assume wlog that two circles pin the north and the west edges of a CSR.
The cases where the two circles define a different pair of adjacent edges of a CSR can be handled in a similar fashion.

If we are in Case 2.1, we consider all corner breakpoints $q$ defined by pairs of adjacent circles in Case D, 
as well as intersection between the south horizontal tangent $t_H$ to the eastmost circle in $B_{NE}$ and the east vertical tangent $t_V$ to the northmost circle in $B_{NW}$ south of $t_H$.
For every such point, the north and west edges are fixed, 
and we either find the south edge by extending the west edge southwards until it hits a blue circle, 
or the east edge by extending the north edge eastwards until it hits a blue circle.
In both approaches, the fourth edge is uniquely determined.
For each circle in $C \in B_{NE}$, we store pointers to the northmost circle in $B_{NW}$ south of $C$ and to the eastmost circle in $B_{SE}$ west of $C$, as well as similar pointers for the other quadrants and directions,
Thus, once the first two edges are fixed, we can find the third and fouth edges in $O(1)$ time. 
Since there are $O(m)$ circles in Case 2.1 and they all can be found in $O(m)$ time, Case 2.1 can be solved in $O(m)$ time.

For Case 2.2, we consider all points $q$ as in Case 2.1. 
Having selected such point $q$ defined by two circles $C_{1} \in B_{NE} \cup B_{NW}$ and $C_{2} \in B_{NW} \cup B_{SW}$, 
we consider the dominating envelope of $B_{SE}$ starting from the east tangent to $C_1$ or the east edge of $S_{min}$, whichever is eastmost, and ending at the south tangent to $C_2$ or the south edge of $S_{min}$, whichever is southmost.
This gives us a range of motion for the $SE$ corner $r$ of the CSR spanning $O(m)$ circle arcs.
For each such arc, we find the optimal CSR in $O(1)$ time as we shall prove in the next section.
Since there are $O(m)$ choices of $q$, we therefore handle Case 2.2 in $O(m^2)$ time.

As for Case 2.3, note that the pairs of circles defining the $NW$ and the $SE$ corners, respectively, must belong to a dominating envelope.
We scan the dominating envelope of $B_{NW}$ for pairs of circles $C_1, C_2$ in Cases B and C and, for each such pair, 
we scan the dominating envelope of $B_{SE}$ for pairs of circles in Cases B and C,
starting from the east tangent to $C_1$ or the east edge of $S_{min}$, whichever is eastmost, and ending at the south tangent to $C_2$ or the south edge of $S_{min}$, whichever is southmost.
Once these pairs are established, the CSR is determined.
Since scanning each dominating envelope takes $O(m)$ time, we handle Case 2.3 in $O(m^2)$ time.

\textbf{Case 3}.

Assume wlog that a circle $C$ pins the east edge of the CSR.
The cases where the circle define a different edge of the CSR can be handled in a similar fashion.

For Case 3.1, we scan the dominating envelope of $B_{NE}$ (similarly, $B_{SE}$) for pairs of circles $(C_{1}, C)$ in cases C and D, ($C$ is the rightmost circle of the pair).
For every such pair of circles in $B_{NE}$, we consider all circles $C_{2} \in B_{SE}$ that are intersected by the west vertical tangent to $C$.
It is possible that some of these circles were already considered for a previous pair of circles in $B_{NE}$, so we may have to consider $O(m^2)$ triplets of circles $(C, C_{1}, C_{2})$.
We also traverse the circles in $B_E$ in increasing X order of their centers.
Denote by $C$ the current circle.
We consider the sequences of circles $C_{NE} \in B_{NE}$ and $C_{2} \in B_{SE}$ that are intersected by the west tangent to $C$.
Since these sequences may include circles already considered for a previous circle in $B_E$, 
we may need to spend $O(m^2)$ to find all such triplets $(C, C_1, C_{2})$ for which there exists a vertical line intersecting both $C_{1}, C_{2}$.
Once a triplet is established, the west, north, and south egdes are established, and the west edge can be determined in $O(1)$ time by extending the north or the south edge until it hits a circle.
Thus, Case 3.1 requires $O(m^2)$ time.

For Case 3.2, we scan the dominating envelope of $B_{SE}$ (similarly, $B_{NE}$) for pairs of circles $(C_{1}, C)$ in cases C and D ($C$ is the rightmost circle of the pair).
We also traverse the circles $C \in B_E$ in increasing X order, and consider the sequences of circles $C_{1} \in B_{SE}$ that are intersected by the west tangent to $C$.
For each pair $(C_{1}, C)$, the east and south edges of the CSR are defined.
This provides a range of eligible circles from the dominating envelope of $B_{NW}$ such that the $SW$ and $NE$ corners of the CSR are not supported by any circle, and the $NW$ corner slides along some circle arc.
There are $O(m)$ $(C_{1}, C)$ pairs and each of them gives $O(m)$ circles from $B_{NW}$.
Hence, Case 3.2 takes $O(m^2)$ time.

\textbf{Case 4}. 

Consider all arcs defined by pairs of adjacent circles in one of the cases A, B, C, or D.
Consider all arcs defined by pairs of adjacent circles.
Each such arc $a$ establishes the range of motion for the appropriate corner $p$ of a CSR, say $[p_{start}, p_{end}]$ in X order.
Suppose $a$ belongs to a circle in $B_{NE}$, which establishes the range of motion of the NE corner $p$ of the CSR.
This gives us a range of sliding motion for the north and the east edges of the CSR, which are supported by two rays $r_W, r_S$ shooting from $p$ to the west and south, respectively.
Since only the SW corner $q$ may also slide along a circle, the west edge can be neither to the west of the first intersection $W(p)$ between $r_W$ and a circle, nor to the west of the easternmost point in $B_W$ of a circle.
Similarly, the south edge can be neither be to the south of the first intersection $S(p)$ between $r_S$ and a circle, nor to the south of the northernmost point in $B_S$ of a circle.
This gives a range of motion for $q$, which may span multiple arcs of circle with X coordinates within the range 
$arcs(p_{start}, p_{end})$ = $[\min(X(W(S(p_{end})))$, $X(S(W(p_{end}))))$, 
$\max(X(W(S(p_{start})))$, $X(S(W(p_{start}))))]$.
In fact, there are $O(m)$ arcs in the worst case, yielding $O(m^2)$ pairs of arcs for all possible pairs $(p, q)$.
By computing the pointers $west(p), south(p), east(p)$, and $north(p)$ for every $p$, from the dominating envelope, we can find each pair of arcs in $O(1)$ time, as we take them in X order.
That is, we consider all arcs $[p_{start}, p_{end}]$ defined by pairs of adjacent circles in X order and, for each such arc, 
we compute the points $W(S(p_{start}))$, $S(W(p_{start}))$, $W(S(p_{end}))$, and $S(W(p_{end}))$, 
and then consider the arcs in $B_{SW}$ with X coordinates within $arcs(p_{start}, p_{end})$.

In the next subsection, we also show how to handle each pair of arcs in Case 4 in $O(1)$ time, in order to find the optimal solution in $O(m^2)$ time.

\subsection{Finding the CSR once the arcs pinning its corners are selected}
\label{AAR-C-fixed-corners-circles}

\begin{figure}[htp]
	\centering
	\includegraphics[scale=0.3]{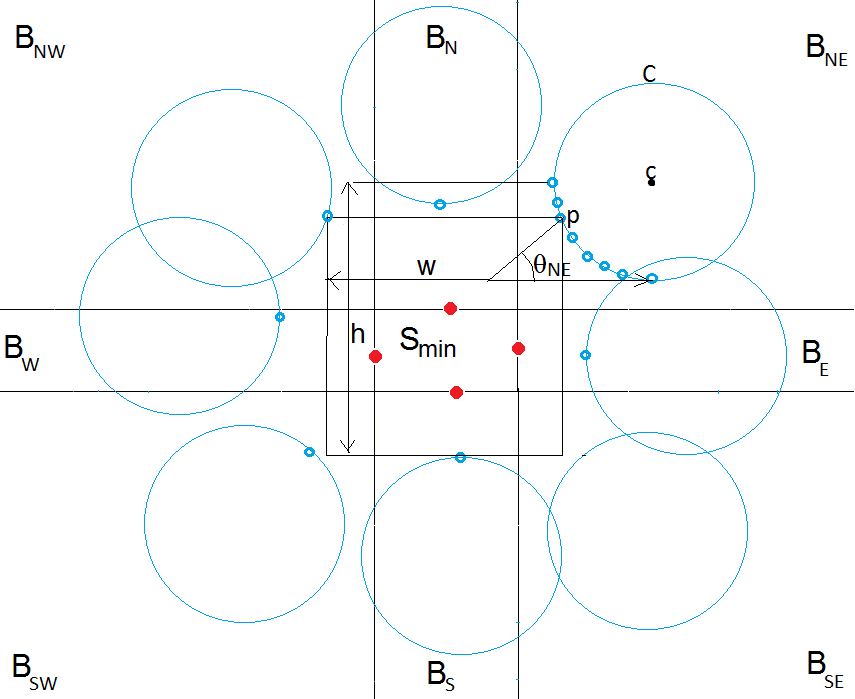}
	\caption{The function $f_{NE}(\theta_{NE})$ denoting the area of the CSR cornered on circle $C(c, 1)$}
	\label{fig:circles-approx}
\end{figure}

We show how to find a maximum separating rectangle once the arcs pinning the corners are selected.

For each quadrant $Q$, let $\theta_Q \in [\alpha_Q, \beta_Q]$ be angular position of the corner within the arc belonging to $Q$, say $C(c, 1)$.
Let $f(\theta_{NE}, \theta_{NW}, \theta_{SW}, \theta_{SE})$ denote the area of the CSR with the corners defined in terms of $theta_Q$ as above.
Our goal is to find the maximum of $f$ over the feasible set of arguments, along with its arguments.
First, assume $\theta_{SE}, \theta_{NW}, \theta_{SW}$ are fixed, with the left and bottom supports denoted as $l, b$ (Figure \ref{fig:circles-approx}).
We refer to $f(\theta_{NE}, \theta_{NW}, \theta_{SW}, \theta_{SE})$ as simply $f_{NE}(\theta)$ (that is, refer to $\theta_{NE}$ as simply $\theta$).

We have the following lemma.

\begin{lemma}
\label{lemma:2.1}
$f_{NE}$ has at most 3 maxima.
\end{lemma}

We have 
\begin{equation}
f_{NE}(\theta) = (w - \sin \theta) \cdot (h - \cos \theta), 
\end{equation}
where $w = x(c) - x(l)$ and $h = y(c) - y(b)$. 
For simplicity, assume that all circles are fully contained in some quadrant, so $w, h > 1$.
Also, 
\begin{equation}
f_{NE}'(\theta) = w \sin \theta - h \cos \theta + \sin^2 \theta - cos^2 \theta.
\end{equation}
Letting $x = \tan{\theta}$, we get 
\begin{equation}
f_{NE}'(x) = \frac{w x - h}{\sqrt{1+x^2}} + \frac{x^2-1}{1+x^2},
\end{equation}
so s$f_{NE}'(x) = 0 \iff (w x - h)\sqrt{1+x^2} = 1 - x^2 \iff
(w x - h)^2 (1+x^2) = (1 - x^2)^2 \iff
(w^2 - 1)x^4 - 2whx^3 + (w^2 + h^2 + 2)x^2 - 2whx + h^2 - 1 = 0 \iff \\
(w^2 - 1)x^2(x^2 - 1) - 2whx(x^2 - 1) + (h^2 - 1)(x^2 - 1) = 0 \iff \\
((w^2 - 1)x^2 - 2whx + h^2 - 1)(x^2 - 1) = 0$,
which solves to 
\begin{equation}
x_1 = 1, 
\end{equation}

\begin{equation}
x_{2,3} = \frac{wh +- \sqrt{2(w^2+h^2) - 1}}{w^2-1}
\end{equation}
 (we ignore negative roots since $x \geq 0$).
Since $x = \tan \theta$, it follows that $\theta_1 = \frac{\pi}{4}, \theta_{2,3} = \arctan{x_{2,3}}$ are extrema for $f_{NE}$.
This means there are at most three maxima for $f_{NE}$.

By a similar argument, there are at most three maxima for $f_{NW}, f_{SW}, f_{SE}$.
Note that the position of two opposite corners of a CSR, say $NE$ and $SW$, determine the position of the other two corners.
This gives us at most $M = 9$ maxima for $f$, as there are only two variables.
We compute the $O(1)$ maxima of $f$ and choose the one that gives the largest CSR.

Thus, we have proved the following result.

\begin{theorem}
\label{theorem:2.2}
Given a set of red points $R$ with $|R| = n$, and a set of blue unit circles $B$ with $|B| = m$, 
the largest rectangle enclosing $R$ and avoiding all circles in $B$ can be found in $O(m^2 + n)$ time.
\end{theorem}

\section*{Acknowledgement}
\label{ack}

The author would like to thank Drs. Benjamin Raichel, Chenglin Fan, and Ovidiu Daescu for the useful discussions.

\section{Conclusions and Future Work}
\label{Conclusions and Future Work}

We consider the outlier version of the largest axis-aligned separating rectangle (MBSR-O),
for which we give an $O(k^3 m + m \log m + n)$ time algorithm.
We also study the problem of finding the largest axis-aligned separating rectangle among unit circles (MBSR-C) and give an $O(m^2 + n)$ time algorithm.

We leave for future consideration finding the largest, as well as the smallest enclosing circle avoiding all blue circles.
A "combined" version such as MBSR-C with outliers, i.e., finding the largest rectangle enclosing all red points while containing at most $k$ circles, would also be of interest.
Finally, it would be interesting to either further improve the time bounds for the MBSR-O and MBSR-C, prove lower bounds, or come up with approximation algorithms.


\small
\bibliographystyle{abbrv}

\end{document}